\documentclass[11pt,aip]{revtex4-1}

\usepackage{bm, amsfonts, mathtools}
\usepackage{graphicx, color, wasysym}
\usepackage{textcomp}
\usepackage{amsmath, amssymb, amsthm, latexsym, epsfig}
\usepackage{appendix}

\usepackage{stackrel}

\usepackage{txfonts}

\DeclareMathOperator{\okr}{{\stackrel{{\scriptscriptstyle{\mathsf{def}}}}{=}}}
\DeclareMathOperator{\D}{d\!}

\theoremstyle{plain}
\newtheorem*{thmm}{Theorem}
\newtheorem*{exaa}{Example}
\numberwithin{equation}{section}
\theoremstyle{remark}
\newtheorem*{remark}{{\it Remark}}

\begin{document} 


\makeatletter

\title{A note on the article "Anomalous relaxation model based on the fractional derivative with a Prabhakar-like  kernel" [Z. Angew. Math. Phys. (2019) 70: 42]}

\author{K. G\'{o}rska}
\author{A. Horzela}

\affiliation{H. Niewodnicza\'{n}ski Institute of Nuclear Physics, Polish Academy of Sciences,  ul.Eljasza-Radzikowskiego 152, PL 31342 Krak\'{o}w, Poland \vspace{2mm}}

\author{T. K. Pog\'{a}ny}
\affiliation
{Faculty of Maritime Studies, University of Rijeka, Studentska 2, HR-51000 Rijeka, Croatia\\ Institute of Applied Mathematics, \'Obuda University, B\'ecsi \'ut 96/b, H-1034 Budapest, Hungary}

\email{katarzyna.gorska@ifj.edu.pl; andrzej.horzela@ifj.edu.pl; poganj@pfri.hr}

\begin{abstract}
Inspired by the article "Anomalous relaxation model based on the fractional derivative 
with a Prabhakar-like kernel" (Z. Angew. Math. Phys. (2019) 70:42) which authors D. Zhao  and HG. Sun studied the integro-differential equation with the kernel given by the Prabhakar function  $e^{-\gamma}_{\alpha, \beta}(t, \lambda)$ we provide the solution to this equation which is complementary to that obtained up to now. Our solution is valid for effective relaxation times which admissible range extends the limits given in \cite[Theorem 3.1]{DZhao2019} to all positive values. For special choices of parameters entering the equation itself and/or characterizing the kernel the solution comprises to known  phenomenological relaxation patterns, e.g. to the Cole-Cole model (if $\gamma = 1, \beta=1-\alpha$) or to  the standard Debye relaxation. 
\end{abstract}


\maketitle

\section{Introduction}

In the recently published article \cite{DZhao2019} its authors Dazhi Zhao and HongGuang Sun 
studied the linear integro--differential equation
   \begin{equation}\label{30/06-1}
      \int_{0}^{t} e^{-\gamma}_{\alpha, \beta}(t-t^{\prime}, \lambda) \frac{\D}{\D t'} 
			f(t^{\prime}) \D t' = - M(\tau, \alpha) f(t) 
   \end{equation}
where the kernel $k(t; \alpha) = e^{-\gamma}_{\alpha, \beta}(t; \lambda)$ is given by the 
Prabhakar function which parameters satisfy $0<\gamma\leq 1$ and $\alpha, \beta > 0$, 
$\alpha + \beta = 1$. For this range of parameters recall that the Laplace transform 
of  $k(t; \alpha)$,  namely $K(s, \alpha) = s^{-\alpha\gamma-\beta}(s^{\alpha}-\lambda)^{\gamma}$, 
satisfies the condition $\lim_{s\to\infty}[s K(s, \alpha)]^{-1} = 0$, which according to 
\cite[Eq. (2) {\it et seq.}]{DZhao2019} permits to qualify the integro-differential operator in 
Eq. \eqref{30/06-1} as the so-called generalized Caputo (GC) derivative. Here $M(\tau, \alpha)$ stands for  
$\Lambda(\tau, \alpha)/N(\alpha)$ where $N(\alpha) = (1-\alpha)^{-1}$ normalizes the integral in 
Eq. \eqref{30/06-1} and $\Lambda(\tau, \alpha)$ is a function of the effective relaxation time $\tau$.

Considering Eq. \eqref{30/06-1} as a model of  the anomalous relaxation and solving it the authors of 
\cite{DZhao2019} showed that the model extends the Cole--Cole relaxation pattern and contains as 
the limiting case $\alpha\to 1$ the standard Debye relaxation. Here we would like to emphasize that just mentioned two cases do not exhaust possible mutual relations which link the relaxation phenomena and  using the Eq. \eqref{30/06-1} for  modeling their time behavior. 
An instructive example is an application of Eq. \eqref{30/06-1}-like equation to describe the Havriliak--Negami relaxation, the most widely used "asymmetric" generalization of the Debye and Cole--Cole approaches. In the review paper \cite{FMainardi} the  
authors presented a detailed analysis of equations describing the time behavior of the 
Havriliak--Negami relaxation function $\Psi_{\alpha, \gamma}(t)$. They came to the conclusion that it  
is governed by a non-homogenous equation
   \[ {^{C}(_{0}D_{t}^{\alpha} + \tau^{-\alpha})}^{\gamma} \Psi_{\alpha, \gamma}(t) 
			    = -\tau^{-\alpha\gamma}, \qquad \Psi_{\alpha, \gamma}(0) = 1, \]
where the pseudo-differential operator ${^{C}(_{0}D_{t}^{\alpha} + \tau^{-\alpha})}^{\gamma}$ is a 
Caputo-like counterpart of the operator $(_{0}D_{t}^{\alpha} + \tau^{-\alpha})^{\gamma}$, the latter 
understood as an infinite binomial series of the Riemann-Liouville fractional derivatives
\footnote{For a comprehensive information about ${^{C}(_{0}D_{t}^{\alpha} + \tau^{-\alpha})}^{\gamma}$ 
see \cite[Section 3.3, Appendix B]{FMainardi}.}. Next, using results of \cite{RGarra14}, 
they argued that the operator ${^{C}({_{0}D_{t}}^{\alpha} + \tau^{-\alpha})}^{\gamma}$ may be 
represented in terms of an integro-differential operator involving the Prabhakar function in the 
kernel, the object usually nick-named the Prabhakar derivative. Adjusted to our notation the suitable 
equations \cite[Eq. (B.23)]{FMainardi} read 
   \begin{align*}
     \begin{split}
                {^{C}(_{0}D_{t}^{\alpha} + \tau^{-\alpha})}^{\gamma} \Psi_{\alpha, \gamma}(t) 
				&\equiv e^{-\gamma}_{\alpha, 1 - \alpha\gamma}(t; \lambda) \star 
				        \frac{\D}{\D t} \Psi_{\alpha, \gamma}(t) \\
        & = \int_{0}^{t} e^{-\gamma}_{\alpha, 1-\alpha\gamma}(t-u; \lambda) 
				       \Psi^{\prime}_{\alpha, \gamma}(u) \D u,
     \end{split}
   \end{align*} 
where $\star$ denotes the convolution operator. 
This justifies the condition $\beta = 1-\alpha\gamma$ to appear in Eq. \eqref{30/06-1} as meaningful for 
understanding properties of physically admissible relaxation models. In \cite{RGarra18} it has been also shown 
that the nonlinear heat conduction equations with memory involving Prabhakar derivative can 
be characterized by Eq. \eqref{30/06-1} in which $\beta = 1-\alpha\gamma$. 

The Laplace transform method applied to Eq. \eqref{30/06-1}  results in $F(s) =  f(0+) H(s, \alpha)$, where 
   \begin{equation}\label{5/07-3}
      H(s, \alpha) = \dfrac{K(s, \alpha)}{s K(s, \alpha) + M(\tau, \alpha)},
   \end{equation}
in which the inverse Laplace transform of $F(s)$, denoted as $f(t)$, satisfies 
$\lim_{t \to \infty} f(t) < \infty$. In what follows 
   \begin{equation} \label{constraint}
	    f(0+) \equiv 1 
	 \end{equation} 
will be used throughout, since this constraint neither harms nor restricts our further considerations. 
In \cite{DZhao2019} the authors used the fact that the inverse Laplace transform of the geometric series (which results after 
pulling out $K(s, \alpha)$ in the nominator and denominator of Eq. \eqref{5/07-3} and subsequently 
reducing it) may be performed termwise. This leads to their main result formulated as 
\cite[Theorem 3.3]{DZhao2019}
   \begin{equation}\label{5/07-4}
      f(t) = \sum_{r\geq 0} (-1)^{r} M^r(\tau, \alpha)\, 
			       e^{r\gamma}_{\alpha, 1+ r(1-\beta)}(t; \lambda)\,, 
   \end{equation}
for $|M(\tau, \alpha)/[s K(s, \alpha)]| < 1$, bearing in mind Eq. \eqref{constraint}. The aim of our 
note is to show that just given restriction is not mandatory to solve the Eq. \eqref{30/06-1} as we 
can consider the inverse Laplace transform of Eq. \eqref{5/07-3}, namely the function $f(t)$, also 
for $|M(\tau, \alpha)/[s K(s, \alpha)]| > 1$. 

The note is organized as follows: we begin with a few less known remarks on 
the properties of the Prabhakar function with negative upper index, next show how to find the 
solution for  $|M(\tau, \alpha)/[s K(s, \alpha)]| > 1$ and complete the paper with remarks 
concerning relations between the standard Cole--Cole model and the solution to the Eq. \eqref{30/06-1}. 
We also comment how the results of  \cite{DZhao2019} and this work are viewed in the light of 
general approach proposed in \cite{Kochubei2011}.

\section{The Prabhakar function}

The Prabhakar function \cite{TRPrabhakar69}
   \begin{equation}\label{27/07-2}
      e^{\gamma}_{\alpha, \beta}(t, \lambda) \okr t^{\beta - 1} E^{\gamma}_{\alpha, \beta}
			           (\lambda t^{\alpha})
   \end{equation}
is expressed by the three parameters Mittag-Leffler function $E^{\gamma}_{\alpha, \beta}
(\lambda t^{\alpha})$ defined by the series \cite[p. 7, Eq. (1.3)]{TRPrabhakar69}
    \[ E^{\gamma}_{\alpha, \beta}(x) = \sum_{r \geq 0}\dfrac{(\gamma)_r\, x^r}
		                                {r!\, \Gamma(\alpha r + \beta)}, 
																		\qquad \Re(\alpha)>0;\, \beta, \mu \in \mathbb C;\]
here $(\gamma)_r = \Gamma(\gamma+r)/\Gamma(\gamma)$ stands for the familiar Pochhammer symbol. 
If $\gamma = -n$, $n$ positive integer, the three parameter Mittag--Leffler function is given through 
hypergeometric type polynomial
   \begin{equation}\label{27/07-1}
      E^{\,-n}_{\alpha, \beta}(x) = \dfrac{1}{\Gamma(\beta)} \sum_{k=0}^{n} 
				                            \dfrac{(-n)_{k}}{(\beta)_{\alpha k}} \dfrac{x^{k}}{k!} 
                                  = \dfrac1{\Gamma(\beta)}\, {}_1\Psi_1 \Big[ \begin{array}{c} (-n, 1)\\ 
						                        (\beta, \alpha) \end{array} \Big| \, x \Big]. 
   \end{equation}
For positive integer $\alpha$ they are the biorthogonal polynomials pairs  discussed in 
\cite{TRPrabhakar69, JDEKonhauser67, HMSrivastava82}; the polynomials with general 
values of $\alpha > 0$ are mentioned in \cite{RGarra18}. Here ${_1\Psi_1}$ stands for the confluent 
generalized hypergeometric function, see for instance \cite[p. 21]{SriKar}. The particular case 
of Eq. \eqref{27/07-1} for $n=1$ reads
   \begin{equation}\label{27/07-3}
      E^{-1}_{\alpha, \beta}(x) = \dfrac1{\Gamma(\beta)} + \dfrac{x}{\Gamma(\alpha + \beta)}.
   \end{equation}
This expression will be used in the {\em Remark} which closes the next section and 
enables a comment on the relation between Eq. \eqref{30/06-1} and the Cole--Cole relaxation model.

\section{Alternative solution of Eq. \eqref{30/06-1}}

As previously mentioned the case when $|M(\tau, \alpha)/[s K(s, \alpha)]| > 1$ has not 
been included in considerations presented in  \cite{DZhao2019}. To fill this gap we shall proceed in 
an analogous way and formulate 

\begin{thmm}\label{t5/07-1}
For $|M(\tau, \alpha)/[s K(s, \alpha)]| > 1$ the solution of {\rm Eq.} \eqref{30/06-1} becomes
   \begin{equation}\label{5/07-5}
      f(t) =  \dfrac1{M(\tau, \alpha)} \sum_{r \geq 0} \dfrac{(-1)^r}{M^r(\tau, \alpha)}\, 
			        e^{-(1+r)\gamma}_{\alpha, 1- (1+r)(1-\beta)}(t; \lambda) .
   \end{equation}
\end{thmm}

\begin{proof}
First we pull out $M(\tau, \alpha)$ in the denominator of $H(s, \alpha)$ given by Eq. \eqref{5/07-3}. 
Thus it can be rewritten in the form
   \begin{equation}\label{5/07-6}
      H(s, \alpha) = \frac{K(s, \alpha)}{M(\tau, \alpha)} \left[1 
			             + \frac{s K(s, \alpha)}{M(\tau, \alpha)}\right]^{-1}.
   \end{equation}
Next, after applying the series expansion of $(1+x)^{-1}=\sum_{r\geq 0}(-x)^{r}$ for $|x|<1$, the 
Eq. \eqref{5/07-6} with $x = s K(s, \alpha)/M(\tau, \alpha)$ can be expressed as
   \begin{equation}\label{5/07-7}
      H(s, \alpha) = \sum_{r\geq 0} (-1)^{r} M^{-1-r}(\tau, \alpha) s^{r} K^{1+r}(s, \alpha).
   \end{equation}
The condition $|x|<1$ means that $|M(\tau, \alpha)/s K(s, \alpha)| > 1$. Substituting the explicit 
form of $K(s, \alpha)$ given below Eq. \eqref{30/06-1} into Eq. \eqref{5/07-7} we obtain Eq. 
\eqref{5/07-5}, as $f(0+) = 1$. That finishes the  proof.
\end{proof}

\begin{exaa}
{\rm Taking the same values of parameters $M(\tau, \alpha)$ and $\gamma = 1$ as in \cite[p. 42, 
{\em Example} 3.4]{DZhao2019} {the constraint 
   \[ |M(\tau, \alpha)/[s K(s, \alpha)]| > 1\]}
used to get \eqref{5/07-5} gives different, but 
complementary restriction on $\tau$ from that found in \cite{DZhao2019}. Namely, we get 
$\tau< (1-\alpha)^{2}/(b \alpha)$ while in \cite{DZhao2019} one finds $\tau> (1-\alpha)^{2}/(b \alpha)$; 
both conditions merged together cover the admissible range of $\tau$. To provide numerical estimations 
we take  $b=1$, $\alpha = 0.5$ and $\alpha = 0.7$ which leads to  $\tau <1/2$ and $\tau < 9/70$, 
respectively. This means that with growing $\alpha$ our solution \eqref{5/07-5} works for shorter 
and shorter characteristic relaxation times $\tau$'s, while for $\alpha$ close to $0$ it covers almost 
all range of $\tau$. 
$\blacksquare$}
\end{exaa}

For the values of parameters listed in the example above, i.e. 
$\gamma = 1$, $M=(1-\alpha)/\tau$, $\lambda = -b \alpha/(1-\alpha)$, and $K(s, \alpha) = 
s^{-1} (s^{\alpha} - \lambda)$, the Eq. \eqref{5/07-7} reads
   \begin{align}\label{12/07-0}
      H(s, \alpha) &= \frac{s^{\alpha} - \lambda}{s M(\tau, \alpha)} \sum_{r\geq 0} 
			                \left[-\frac{s^{\alpha} - \lambda}{M(\tau, \alpha)}\right]^{r} \nonumber \\
                   &= \frac{s^{\alpha - 1}}{s^{\alpha} + M(\tau, \alpha) - \lambda} 
									  - \frac{\lambda s^{-1}}{s^{\alpha} + M(\tau, \alpha) - \lambda}
   \end{align}
which is satisfied for $\tau < (1-\alpha)^{2}/(b\alpha)$. The same results can be obtained by 
using Eq. \eqref{5/07-4}, i.e. \cite[Theorem 3.1]{DZhao2019}, but, now, for 
$\tau > (1-\alpha)^{2}/(b\alpha)$. This suggest that to have Eq. \eqref{12/07-0} satisfied 
we do not need to put any additional constraint on $\tau$ except of its positivity. Indeed, 
Eq. \eqref{5/07-3} valid for $\tau > 0$ is equal to Eq. \eqref{12/07-0}. Hence, from the Laplace 
transform of the three parameters Mittag-Leffler function (recalling that $f(0+)=1$)  we conclude 
   \begin{equation}\label{12/07-1}
      f(t) = E_{\alpha}\big(-[M(\tau, \alpha) - \lambda] t^{\alpha}\big) 
           - \lambda t^{\alpha} E_{\alpha, 1 +\alpha} \big(-[M(\tau, \alpha) - \lambda] t^{\alpha}\big)\,,
   \end{equation}
which, after using the suitable property of the Mittag--Leffler functions (see 
\cite[Eq. (4.2.3)]{ML2014}) implies 
   \begin{equation}\label{12/07-2}
      f(t) = \dfrac{M(\tau, \alpha)}{M(\tau, \alpha) - \lambda} E_{\alpha}\big(-[M(\tau, \alpha) 
			     - \lambda] t^{\alpha}\big) - \dfrac{\lambda}{M(\tau, \alpha) - \lambda}\,.
   \end{equation}
Thus, \cite[Eq. (19)]{DZhao2019} can be treated as the approximation of exact solution given by 
Eq. \eqref{12/07-1} or Eq. \eqref{12/07-2}.

\begin{remark}
Eq. \eqref{30/06-1} for $\gamma = 1$ in which we applied Eqs. \eqref{27/07-2} and 
\eqref{27/07-3} can be written as
   \[ {^{C}\!D_{t}^{1-\beta}}f(t) + \lambda\cdot {^{C}\!D_{t}^{1-\alpha-\beta}}f(t) 
			     = - M(\tau, \alpha) f(t), \]
{where for an $\eta$ suitable, 
   \[ {^{C}\!D_{t}^{\eta}}f(t) = \dfrac1{\Gamma(1-\eta)} \int_{0}^{t} (t-u)^{-\eta} f'(u) \D u \] 
stands for} the Caputo fractional derivative. For $\beta = 1 -\alpha$ we get
   \begin{equation}\label{12/07-4}
      {^{C}\!D_{t}^{\alpha}}f(t) + [M(\tau, \alpha) - \lambda]f(t) = \lambda f(0+) = \lambda\,,
   \end{equation}
whose solution coincides with \eqref{12/07-1}, see e.g. \cite{IPodlubny99, BJWest10, KGorska12}. 
For $\lambda = 0$ the Eq. \eqref{12/07-4} becomes the equation relevant for the Cole-Cole 
relaxation. Simultaneously, we have the relation $e^{-1}_{\alpha, 1 -\alpha}(t; 0) = 
\frac{t^{-\alpha}}{\Gamma(1-\alpha)}$ easily seen from the Eq. \eqref{27/07-3} for $\lambda = 0$. 
It implies that the Prabhakar derivative becomes Caputo fractional 
derivative and Eq. \eqref{30/06-1} tends to the evolution equation describing the Cole--Cole 
relaxation. 
\end{remark}

\section{Conclusion}

We would like to point out that our result is complementary to the result given in 
\cite[Theorem 3.1]{DZhao2019} and extends it to the full range of $\tau>0$. This places it 
within the general scheme developed by A. N. Kochubei \cite{Kochubei2011} who 
investigated the Cauchy problem for evolution equations
   \begin{equation}\label{25/07-2}
      (D^{GC}_{t} f)(t) = -M(\tau, \alpha) f(t).
   \end{equation} 
governed by the integro-differential operator 
   \[ (D^{GC}_{t}f)(t) = \frac{\D}{\D \tau}\int_{0}^{t} k(t-\tau, \alpha) 
			                   f(\tau) \D\tau - k(t)f(0).\]
In addition some requirements are put on the Laplace transform $K(s,\alpha)$ of the kernel 
$k(t,\alpha)$. {Namely, it belongs to the Stieltjes class and satisfy the 
following asymptotic conditions: if $s\to 0$ then $K(s,\alpha) \to \infty$ and $sK(s,\alpha) \to 0$, 
while in the case $s \to \infty$,  
there hold $K(s,\alpha)\to 0$ and $sK(s,\alpha)\to\infty$. For instance, under this study all these 
conditions are satisfied and according to} \cite[Theorem 2]{Kochubei2011} the solution $f(t)$ is continuous 
on $[0,\infty)$, infinitely differentiable and completely monotone on $(0,\infty)$. 

Physical usefulness of the Eq. \eqref{30/06-1} as a tool to develop a description of the 
anomalous relaxation patterns is rooted in its relation to the Cole-Cole and Debye models. 
The first case has been just discussed  in the above.  The Debye relaxation emerges when 
$K(s, \alpha)$ is a constant and consequently $k(t) = B(\alpha) \delta(t)$. It is seen 
from
   \begin{equation}\label{2/03-3}
      H(s, \alpha) = \frac{B(\alpha)}{\Lambda(\tau, \alpha)} \left[1 + s \frac{B(\alpha)}
			              {\Lambda(\tau, \alpha)}\right]^{-1},
   \end{equation}
obtained either from \eqref{5/07-4} or \eqref{5/07-5}. Calculating the inverse Laplace transform 
of \eqref{2/03-3} we obtain the solution of Eq. \eqref{25/07-2} in the form
   \[ f(t) = \exp\left[- \frac{\Lambda(\tau, \alpha)}{B(\alpha)} t\right] ,\]
which is the Debye relaxation function in time domain. 

\section*{Acknowledgements} The research of K. G. and A. H. was supported by the Polish 
National Center for Science (NCN) research grant OPUS12 no. UMO-2016/23/B/ST3/01714. K. G. and 
T. K. P acknowledge the support of NAWA (National Agency for Academic Exchange, Poland):  K.G. in 
the framework of the  Bekker Project (PPN/BEK/2018/1/00184) which provided her the opportunity to complete this work during 
the stay in the ENEA Research Center Frascati while T. K. P. under the project 
PROM PPI/PRO/2018/1/00008. T. K. P. also thanks the INP PAS for the warm hospitality during his 
stay in Krak\'{o}w, Poland.

\end{document}